\documentclass[11pt]{article}
\usepackage{geometry}
\usepackage[english]{babel}
\usepackage[utf8]{inputenc}
\usepackage[T1]{fontenc}
\usepackage{indentfirst}
\usepackage{amsmath}
\usepackage{amssymb}
\usepackage{amsthm}
\usepackage{proof}
\usepackage{eufrak}
\usepackage{subfigure}
\usepackage{psfrag}

\allowhyphens

\newcommand{\complex}{\mathbb{C}}

\newcommand{\valos}{\mathbb{R}}

\newtheorem{thm}{Theorem}
\newtheorem{lem}{Lemma}

\newcommand{\ket}[1]{{\left|#1\right\rangle}}
\newcommand{\bra}[1]{{\left\langle #1\right|}}

\newcommand{\AAA}{\mathcal{A}}
\newcommand{\DD}{\mathcal{D}}
\newcommand{\VV}{\mathcal{V}}


\usepackage[colorlinks=true,urlcolor=blue,citecolor=blue,linkcolor=blue,pdfborderstyle={/S/U/W 0}]{hyperref}

\usepackage{authblk}

\begin{document}

\title{
  Crosscap states with tunable entanglement as exact eigenstates of local spin chain Hamiltonians
}
\author[1]{M\'arton Mesty\'an}
\author[1]{Bal\'azs Pozsgay}

\affil[1]{MTA-ELTE “Momentum” Integrable Quantum Dynamics Research Group \protect\\
   E\"otv\"os Lor\'and University, Budapest, Hungary} 
\maketitle

\begin{abstract}
{  It has been observed recently that various spin chain Hamiltonians admit special zero energy ``crosscap'' eigenstates. These states are made up of maximally entangled Bell pairs prepared on antipodal sites of a periodic chain. 
We generalize the states by allowing the antipodal pairs to have non-maximal, tunable entanglement. We give sufficient conditions for such states to be exact zero energy eigenstates of a local Hamiltonian.
The conditions are naturally satisfied in many models which have a global $U(1)$ symmetry. These models include
well known integrable models such as the XX model, the Bariev model, the folded XXZ model, and also a variety of
non-integrable models. Using the zero-energy crosscap states we also derive a family of exact zero modes with sub-volume law 
entanglement.} 
\end{abstract}

\section{Introduction}

Finding exact eigenstates in strongly interacting models is an area of research in itself.
In generic ergodic models most eigenstates are such that there are no explicit formulas for the coefficients of the
wave function in any local basis. By contrast, there are certain exactly solvable models where the form of the wave function
is known for all of the eigenstates. Examples include models solvable by free fermions and one dimensional integrable models solved
by the Bethe Ansatz \cite{gaudin-book-forditas}.

In the last couple of years multiple types of exact eigenstates were found in non-integrable spin chain
models. Such eigenstates include exact 
many-body scars, which are eigenstates with low entanglement in the middle of the spectrum \cite{hfrag-review,scars-review-2}. Interest in many-body scars is motivated by their exotic behaviour: they lead to a
weak breaking of 
ergodicity,  observable in experimental situations \cite{scars-nature,scars-elso}. 

More recently, further unexpected exact eigenstates have been found in various models. These new types of states cannot be
characterized as quantum many-body scars according to the conventional definitions, because they can have volume law 
entanglement, and their local reduced density matrices can appear to be thermal. One class of such states, dubbed the \emph{rainbow scar} \cite{rainbow-scars}, is built from Bell pairs distributed in a reflection symmetric way,
such that the entanglement pattern can be drawn as a rainbow.

{Very recently an arrangement similar to the rainbow scar has been used
  in various \emph{periodic} spin chain models to build exact eigenstates with zero energy and volume-law entanglement \cite{crosscap-zero-0,crosscap-pxp,yoneta-chiba,crosscap-zero-1}. In these states, called \emph{Entangled Antipodal Pair States} \cite{yoneta-chiba}, every spin forms a maximally entangled Bell pair with its antipodal partner along the circle. The resulting states behave as an infinite temperature
ensemble for all local correlation functions  \cite{yoneta-thermal-pure}. 
  The existence of EAPS eigenstates often appears together with an exponential
    degeneracy of the zero energy subspace but this is not necessarily the case \cite{yoneta-thermal-pure}.}

    The idea of entangling antipodal pairs of a
    1+1 dimensional system appeared earlier in 
       integrable spin chain models and conformal field theories \cite{crosscap-shota}. The resulting states are called crosscap states. 
In integrable spin chains a key property of crosscap states is that they
are {\it integrable states} 
\cite{sajat-integrable-quenches}.  Integrable states are not eigenstates of the Hamiltonian, but they are annihilated by
those higher charges which are odd under space reflection. Integrable states are distinguished because
their exact overlaps with the eigenstates are known in the form of determinant formulas \cite{sajat-neel,sajat-minden-overlaps,tamas-all-twosite-overlaps,tamas-all-mps-overlaps}. These formulas enable the
analytic treatment of real time evolution  starting from a crosscap state \cite{JS-oTBA,sajat-oTBA}.
Integrable crosscap 
states and their overlaps were studied in detail in \cite{crosscap-tamas}, and quantum quenches from crosscap states
were treated in \cite{crosscap-quench}.  Crosscap states were studies also in the Lieb-Liniger model (interacting 1D
Bose gas) \cite{yunfeng-crosscap-LL} and also in classical field theories \cite{tamas-crosscap-classical}.

An important feature of integrable crosscap states is that their entanglement is tunable because the anti-podal pairs are not
necessarily 
maximally entangled \cite{crosscap-tamas}. This raises the question, whether \emph{zero-energy} crosscap states with tunable
entanglement can exist in Hamiltonians of known non-integrable or integrable models.
In the very recent work \cite{crosscap-tunable} exact zero modes with tunable entanglement were found in certain special
models, which have a staggering in the sign of their spin-spin interaction terms. The states are expressed as symmetric
combinations of selected Bell pairs prepared on antipodal sites. 

{Motivated by the recent works \cite{yoneta-chiba,crosscap-zero-1,crosscap-tunable} we ask the question: which local Hamiltonians admit
exact zero modes constructed from tunable anti-podal pairs? In the case of non-tunable, maximally entangled pairs a sufficient set of
criteria was found in \cite{yoneta-chiba} for spin-1/2 models and in \cite{crosscap-zero-1} for models with arbitrary local dimension but a restricted form of the Hamiltonian. We derive a 
sufficient condition for the existence of translationally invariant crosscap zero states with tunable entanglement, valid for arbitary local dimension. Our computation is very straightforward, 
nevertheless it appears that our main result was not yet written down in its full generality.}

The organization of the paper is as follows: the main result is presented in Section \ref{sec:main}, together with a discussion of the potential exponential degeneracies of
the null spaces. Afterwards in Section \ref{sec:U1} we show that models with $U(1)$ symmetries often admit solutions to
the main condition, if the hopping terms in the Hamiltonian have a particular anti-symmetric form. We also show how to
modify our construction so that it applies also to models with symmetric hopping amplitudes. The integrability or
non-integrability of the model does not play any role, but we find restrictions for the allowed interactions in the
Hamiltonian density. Afterwards in Section \ref{sec:integr} we discuss connections with the integrable crosscap states
of integrable spin chains. There we also discuss connections to the construction of \cite{crosscap-tunable}. Finally we
present our Conclusions in Section \ref{sec:concl}.

\section{Generalities and main result}

\label{sec:main}

We consider one dimensional spin chains with local Hamiltonians and periodic boundary conditions. The local Hilbert
space is $\complex^D$ with some $D\ge 
2$. Basis states will be denoted as $\ket{a}$ with $a=0, 1, \dots, D-1$.

The full Hilbert space is the $2L$-fold tensor product of $\complex^D$, where $L$ is chosen to be the half of the length of the
chain. The Hamiltonian is written as
\begin{equation}
  H=\sum_{j=1}^{2L} \alpha_j h(j),\qquad \alpha_j\in\valos\,,
  \label{Ham}
\end{equation}
where $h(j)$ is the Hamiltonian density localized around site $j$. We assume that $h(j)$ is the translated copy of the
local operator $h(1)$, and that the support of $h(1)$ is the segment  $[1\dots r]$ of the chain, where $r$ is the range
of the operator density. 

In this Section we allow for a rather general setting with site-dependent $\alpha_j$, but we require
\begin{equation}
  \label{halfp}
  \alpha_j=\alpha_{j+L}\,.
\end{equation}
However, in the concrete examples in the later Sections we will treat homogeneous chains with $\alpha_j=1$.

Let us define a specific Bell pair prepared on two sites $j$ and $k$ as
\begin{equation}
  \ket{\psi_0}_{j,k}=\frac{1}{D}\sum_{a=0}^{D-1} \ket{a}_j\otimes \ket{a}_k\,.
\end{equation}
We construct a state $\ket{\Psi_0}$ by preparing such Bell pairs on antipodal sites,
\begin{equation}
  \ket{\Psi_0}=\otimes_{j=1}^L \ket{\psi_0}_{j,j+L}\,.
  \label{Psi0}
\end{equation}
 We call this the ``identity crosscap state''.
This state is translationally invariant. A crucial property of this state is that every segment of length smaller than
half the chain is maximally entangled with the rest of the chain.  For such segments the reduced density matrices are
proportional to the identity matrix, therefore the state $\ket{\Psi_0}$ is locally indistinguishable from the infinite
temperature Gibbs state. It is important that this state is not invariant under global rotations, therefore its concrete
representation is basis dependent.

We generalize the construction of $\ket{\Psi_0}$ to the preparation of selected classes of dimers on antipodal
sites. Let us define a 
generic two-site state prepared on sites $j$ and $k$ as
\begin{equation}
  \ket{\psi}_{j,k}=\sum_{a,b=0}^{D-1} A_{ab} \ket{a}_j\otimes \ket{b}_k',.
\end{equation}
We assume that the state is normalized and thus
\begin{equation}
  \sum_{a,b=1}^D |A_{ab}|^2=1 \,.
  \label{norm}
\end{equation}
For reasons that become clear shortly we require that the coefficient matrix $A$ is either symmetric or anti-symmetric:
\begin{equation}
  A_{ab}=A_{ba}\qquad\text{or}\qquad A_{ab}=-A_{ba}\,.
  \label{symm-asymm}
\end{equation}
Below we will often use the object $A$ as a one-site operator, with the convention that for every vector $\ket{v}$ with
coefficients $v_a$ in the local basis the effect of the operator is
\begin{equation}
  (A\ket{v})_a=A_{ab}v_b\,.
\end{equation}
For later use we introduce a shorthand for the product of the $A$ operators on consecutive sites,
\begin{equation}
  \AAA^{(r)}=\prod_{j=1}^r A_j\,.
  \label{mathcalAr}
\end{equation}
We will also use the shorthand
\begin{equation}
  \AAA=\AAA^{(L)}=\prod_{j=1}^L A_j\,.
  \label{mathcalA}
\end{equation}

We define general crosscap states as
\begin{equation}
  \ket{\Psi}=\otimes_{j=1}^{L}  \ket{\psi}_{j,j+L}\,.
  \label{Psi}
\end{equation}
If $A$ is symmetric, then the resulting state is translationally invariant. If $A$ is anti-symmetric, then the state is multiplied by
 $-1$ under one-site translations. The states $\ket{\Psi}$ have volume law entanglement for almost every $A$, however, the amount of
entanglement is typically not maximal.  This can be understood from the physical picture: in a segment of length
$\ell\le L/2$ of the chain every site is entangled with its antipodal point, but the amount of entanglement depends on
the choice of the elements $A_{ab}$. 

Let us now compute the precise values of the von Neumann entanglement entropy in the crosscap states $\ket{\Psi}$ \eqref{Psi}.
First, the reduced density matrix of a single site is 
\begin{equation}
  \rho_1=A^\dagger A\,.
\end{equation}
{Then the reduced density matrix of a segment with length $\ell$ can simply be obtained as 
\begin{equation}
  \rho_{\ell} = \rho_1 ^{\otimes \ell},
\end{equation}
that is, the tensor product of $\ell$ copies of $\rho_1$,} as long as $\ell\le L$. This means that the bi-partite von Neumann entanglement entropy is
proportional to $\ell$ and it is given by
\begin{equation}
  S=-\ell\cdot \text{Tr}\Big[ \rho_1 \log(\rho_1)\Big]=-\ell \cdot \text{Tr}\Big[A^\dagger A \log(A^\dagger A)\Big]\,.
\end{equation}
We obtain maximal entanglement if and only if $A$ is proportional to a unitary matrix.

We note that the construction can also be used to create states with zero entanglement. This is possible when the matrix $A$ has rank
1. In these cases $\ket{\Psi}$ is a product state of one-site vectors.
For example if the matrix $A$ is chosen as
$A_{ab}=\delta_{0,a}\delta_{0,b}$ then the resulting state is simply
\begin{equation}
  \ket{\Psi}=\otimes_{j=1}^{2L} \ket{0}_{j}\,,
\end{equation}
which is a product state with zero entanglement.

Let us now return to the case of a vector $\ket{\Psi}$ is constructed from a generic $A$ matrix satisfying \eqref{norm} and \eqref{symm-asymm}.
We are interested in situations where a local Hamiltonian annihilates such a state:
\begin{equation}
  \label{anni}
  H\ket{\Psi}=0
\end{equation}
The main statement of this paper is the following.
{
\begin{thm}
  \label{maint}
  Consider a periodic spin chain of length $2L$ with a local Hamiltonian $H$ given in \eqref{Ham} whose density $h \equiv h(1)$ has range $r \le L$, and whose couplings satisfy the half-periodicity condition \eqref{halfp}. The crosscap state $\ket{\Psi}$ given in \eqref{Psi}  satisfies the null vector
condition \eqref{anni} with this Hamiltonian if 
\begin{equation}
  \label{cond}
  \AAA^{(r)} h+h^T\AAA^{(r)}=0,
\end{equation}
where the subscript $T$ denotes the usual transposition in the given basis.
\end{thm}}

Equation \eqref{cond} gives a remarkably simple sufficient condition. The steps that lead to
its formulation can 
be found in various places in the literature. 
However, it appears that condition \eqref{cond} in this compact and general form has not been written down elsewhere.

The work \cite{yoneta-chiba} contains a similar relation, but only for the special case of maximally entangled pairs and local dimension 2. It has to be noted that the framework of
\cite{yoneta-chiba} is in a sense more general than ours because it allows for choosing different Bell pairs for each antipodal pair of sites, provided some periodicity rules are satisfied.  
In the recent work \cite{crosscap-zero-1}, essentially the same statement as above has been published for arbitrary local dimension $D$,
but that work only treats maximally entangled pairs and it is restricted to special types of Hamiltonian densities, which do not
exhaust all possibilities for $D>2$.

Now we prove Theorem \ref{maint} using a few simple steps. The proof presented below works when the range of the
Hamiltonian density satisfies $r<L/2$ but it is straightforward to generalize it to $L/2 \le r \le L$. 

We begin with two lemmas.

\begin{lem}
  The crosscap state \eqref{Psi} can be constructed from the ``identity crosscap state'' \eqref{Psi0} as
  \begin{equation}
    \label{cro1}
    \ket{\Psi}=\AAA \ket{\Psi_0}\,.
  \end{equation}
\end{lem}
\begin{proof}
  This follows directly by acting with every $A_j$ on the Bell pair prepared on sites $j$ and $j+L$.
\end{proof}

If $A$ is symmetric, then it follows from translational invariance that in the definition of $\mathcal A$ \eqref{mathcalA} we could place the successive products of $A$
operators along any segment of length $L$ on the chain. If $A$ is anti-symmetric then moving the string of $A$ operators
with one site needs to be compensated by an overall sign. For technical reasons (which will become clear shortly) we
introduce also a different choice for the placement of this operator string. 
Let us choose an even number $m$ such that $r - 1 \le m \le L - r + 1 $ and let us
also define
\begin{equation}
  \tilde\AAA=\prod_{{j=m+1}}^{L+m} A_j
\end{equation}
Using the fact that $m$ is even we obtain an equation analogous to \eqref{cro1}, 
\begin{equation}
  \label{cro2}
    \ket{\Psi}=\tilde \AAA \ket{\Psi_0}\,.
\end{equation}

\begin{lem}
  \label{lem1}
  For any local operator $O(j)$ with range smaller or equal to $L/2$ we have
  \begin{equation}
    \label{push}
    O(j)\ket{\Psi_0}=O^T(j+L)\ket{\Psi_0},
  \end{equation}
  where the subscript $T$ denotes the usual transposition. 
\end{lem}
\begin{proof}
This can be proven directly in the given basis.
\end{proof}

The operation described by eq. \eqref{push} is known more generally as ``operator pushing'', which is possible whenever
one has multi-leg tensors with maximal entanglement along special bi-partitions, see for example
\cite{ads-code-1}.  

After the two lemmas, the next step is to separate the Hamiltonian into two parts, such that each part is a sum of operators acting on antipodal positions:
\begin{equation}
  \label{spli}
  H=\sum_{j=1}^{L-r+1} \alpha_j \Big(h(j)+h(j+L)\Big)+\sum_{j=L-r+2}^L\alpha_j \Big(h(j)+h(j+L)\Big)\,,
\end{equation}
where $r$ is the range of the Hamiltonian densities $h(j)$ \eqref{Ham}.
Let us now act with this Hamiltonian on the crosscap state \eqref{Psi},
\begin{equation}
  \begin{aligned}
  H\ket{\Psi}&=\sum_{j=1}^{L-r+1}  \alpha_j\Big(h(j)+h(j+L)\Big)\ket{\Psi}
  +\sum_{j=L-r+2}^L \alpha_j\Big(h(j)+h(j+L)\Big)\ket{\Psi} \\
  &=\sum_{j=1}^{L-r+1}  \alpha_j\Big(h(j)+h(j+L)\Big)\AAA\ket{\Psi_0}
  +\sum_{j=L-r+2}^L \alpha_j\Big(h(j)+h(j+L)\Big)\tilde\AAA\ket{\Psi_0}\,, \\
\end{aligned}
\label{parteff}
\end{equation}
where in the second line we used \eqref{cro1} and \eqref{cro2}.
In this second line, every term of the first sum is such that $h(j+L)$ acts on sites that are apart from the support of
$\AAA$. Similarly, every term of the second sum is such that $h(j+L)$ acts on sites that are apart from the support of $\tilde \AAA$. This implies that we can  bring $h(j+L)$ to the right of $\AAA$ and $\tilde \AAA$,
and use Lemma \eqref{push} to find
\begin{equation}
  \label{athozva}
  H\ket{\Psi}=\sum_{j=1}^{L-r+1} \alpha_j \Big(h(j)\AAA+\AAA h^T(j)\Big)\ket{\Psi_0}
  +\sum_{j=L-r+2}^L \alpha_j\Big(  h(j)\tilde\AAA + \tilde \AAA h^T(j)  \Big)\ket{\Psi_0}\,.
\end{equation}
Here every term is such that there is an operator of range $L$ acting on $\ket{\Psi_0}$. The expression in \eqref{athozva} implies that the condition \eqref{cond} is sufficient for the vanishing of the r.h.s, and thus Theorem \ref{maint} is proven.

As we have noted, this proof works only when $r<L/2$. However, it is straightforward to generalize the proof to the case
$L/2 \le r \le L$. In this case the bi-partition of the sum in \eqref{parteff} has to be changed to a partition into
more than two partial sums. In each partial sum the operator $\AAA$ has to be translated so that its support does not
overlap with the support of any $h(j+L)$ in the partial sum.

\subsection{Examples}

Let us now discuss simple examples where the condition \eqref{cond} is satisfied.

The simplest possibility is that $A$ is the identity matrix and we are dealing with the state $\ket{\Psi_0}$. In this
case we obtain the simple condition that $h$ should be 
represented by an anti-symmetric matrix. Combined with the Hermiticity of $H$ it also implies that all matrix elements
are purely imaginary.

In the case of qubits the corresponding condition is that $h$ should be a sum of products of Pauli
operators, such that in each term there is an odd number of Pauli $Y$ operators present. This condition was already
given in \cite{yoneta-chiba,crosscap-zero-1}. The statement that for $D>2$ the sufficient and necessary condition is
that $h$ is anti-symmetric appears to be new.

If $h$ is symmetric, then we obtain that it should anti-commute with the tensor product of $A$ operators. An example for this
was treated in \cite{crosscap-zero-0}. The Hamiltonian of  \cite{crosscap-zero-0} has a three site density given by
\begin{equation}
  h=X_1X_2X_3+\gamma Z_2,
\end{equation}
where $\gamma$ is a coupling constant.
This operator density is symmetric in the given basis and it anti-commutes with $\AAA^{(3)}$ if we choose
$A=Y$. This leads to the crosscap states constructed from the anti-symmetric Bell pair, which is a singlet under the
standard action of the $SU(2)$ group. 

In Section \ref{sec:U1} we treat new examples for the solution of the condition \eqref{cond}. But first we discuss
the degeneracy of the null space.

\subsection{Exponential degeneracies}
\label{sub:exp}

It has been observed in previous works that the presence of the crosscap nullstates goes often together with an exponential
degeneracy of the nullspace. Such exponential degeneracies could explain the presence of exotic exact eigenstates in the
spectrum, because taking linear superpositions from a large enough set of vectors can indeed produce unexpected
entanglement properties. 
However,  it has been found in  the work \cite{yoneta-chiba} that the exponential degeneracy is not required:
a concrete model has been found, which supports the
crosscap zero-modes, but which only has a finite amount of nullspace degeneracy.  

Then the question remains: In which models do we observe an exponential degeneracy of the nullspace? This question has been
investigated in various works, see for example \cite{zerospace1,serbyn-area-law}. 
Let us assume that there is an operator $\DD$ which anti-commutes with the Hamiltonian. 
 In such a case the exponential degeneracy of the null space can be
proven using the methods of \cite{zerospace1,serbyn-area-law}, if the Hamiltonian has a point symmetry. In those papers the chosen point symmetry was the spatial reflection symmetry. In the case of a homogeneous periodic chain of length $2L$ that admits crosscap zero states, the natural choice of point symmetry is the translation by $L$ sites. Then the existence
of an anti-commuting operator $\DD$ guarantees the exponential degeneracy of the null space. In the absence of such an
operator the degeneracy of the nullspace can be linear or even constant in system size. 

\section{Crosscap states for models with $U(1)$ symmetries}

\label{sec:U1}

In this Section we explore the consequences of having a conserved $U(1)$ charge. We show that very often a $U(1)$ charge
leads to a one-parameter family of crosscap nullvectors. 

We focus on models with local dimension $D=2$, which have a global $U(1)$ charge generated by the standard
magnetization. Thus the charge $Q$ is given by
\begin{equation}
  \label{Sz}
  Q=\sum_{j=1}^{2L} Z_j\,.
\end{equation}
The key idea to satisfy condition \eqref{cond} is to choose $h$ to be anti-symmetric and to require that it preserves the global
charge:
\begin{equation}
  \label{chargecons}
  [h,Q]=0\,.
\end{equation}
If we choose 
\begin{equation}
  A_j=e^{sZ_j}
\end{equation}
with some complex number $s$, then the commutativity \eqref{chargecons} also implies \eqref{cond}.
It is important that we are free to choose the parameter $s$ to be any complex number.  Below we explore the
consequences of this fact. 

We will see below, that in our examples $h$ is anti-symmetric not only with respect to the transposition, but also with
respect to space 
reflection. For these Hamiltonians we will show that reflection anti-symmetry can be eliminated by special similarity transformations of the Hamiltonian, which will
lead
to new models with a reflection symmetric Hamiltonian. As a by-product we will obtain modified crosscap states for the
new models, which will not be translationally symmetric anymore. However, they will be periodically modulated, similar
to many examples treated in earlier work \cite{crosscap-zero-0,yoneta-chiba,crosscap-zero-1}.

Many examples that we treat are integrable models, but we will show how to add simple
modifications to the Hamiltonians to make them non-integrable, while keeping the condition \eqref{cond} intact.

\subsection{The anti-symmetric XY model}

\label{sec:XY}

Let us take the model defined by the two-site Hamiltonian density
\begin{equation}
  \label{hxy}
  h_{XY}=X_1Y_2-Y_1X_2=2i (\sigma^+_1\sigma^-_2-\sigma^-_1\sigma^+_2),
\end{equation}
where we used the raising and lowering operators defined as
\begin{equation}
  \sigma^+=\ket{1}\bra{0}\qquad \sigma^-=\ket{0}\bra{1}\,.
\end{equation}
Later we will show that if $L$ is even (so the total length $2L$ of
the chain is a multiple of 4) then this model is equivalent to the standard XX model defined by the Hamiltonian density
\begin{equation}
  h_{XX}=X_1X_2+Y_1Y_2=2(\sigma^+_1\sigma^-_2+\sigma^-_1\sigma^+_2)
\end{equation}

We can see that $h_{XY}$ is anti-symmetric with respect to transposition and also with respect to space
inversion. The concrete matrix representation is
\begin{equation}
  h_{XY}=
    \begin{pmatrix}
   & & &\\\
      &  & 2i& \\
      &-2i &  &\\
    & & & 
  \end{pmatrix}\,.
\end{equation}
This Hamiltonian density conserves the $U(1)$ charge
\eqref{Sz}. It follows that we can choose $A$ to be 
{\it any} diagonal matrix. For
example, if we choose
\begin{equation}
  A\sim 
  \begin{pmatrix}
    1 & \\ & c
  \end{pmatrix}\,.
\end{equation}
then the relevant tensor product for condition \eqref{cond} is
\begin{equation}
A^{(2)}=    A\otimes A\sim 
  \begin{pmatrix}
    1 & & &\\\
      & c & & \\
      & & c &\\
    & & & c^2
  \end{pmatrix}\,.
\end{equation}
In these formulas we omitted the unnecessary normalization factors.

It is clear from the matrix representations that $[A^{(2)},h_{XY}]=0$ for every $c$. We can thus construct zero
energy eigenvectors with the parameter  $c$:
\begin{equation}
    \ket{\Psi(c)}=\otimes_{j=1}^{L}  \ket{\psi(c)}_{j,j+L},
    \label{eq:Psic}
\end{equation}
where now we choose an un-normalized dimer
\begin{equation}
  \ket{\psi(c)}=\ket{0}\otimes \ket{0}+c (\ket{1}\otimes \ket{1})\,.
\end{equation}
Theorem \ref{maint} guarantees that vector $\ket{\Psi(c)}$ is a zero energy state of the Hamiltonian for every $c$.

The existence of a one-parameter family of null states has interesting consequences. We 
can take derivatives with respect to $c$ and thus find a new basis within the space spanned by the $\ket{\Psi(c)}$.
We define
\begin{equation}
  \ket{\Psi^{(n)}}=\left.\left(\frac{d}{dc}\right)^n\ket{\Psi(c)}\right|_{c=0}
\end{equation}
The state $\ket{\Psi^{(0)}}$ is identical to the ferromagnetic product state $\ket{\emptyset}$ which is defined as
\begin{equation}
  \ket{\emptyset}=\otimes_{j=1}^{2L}\ket{0}_j
\end{equation}
The states $\ket{\Psi^{(n)}}$ with $n\ge 1$ can be
interpreted as having excitations above the ferromagnetic state. The excitations are antipodal pairs $\ket{1}\otimes\ket{1}$.
The state  $\ket{\Psi^{(n)}}$ has $n$ number of such antipodal pairs distributed in a completely symmetric
way. More concretely it is given by
\begin{equation}
  \ket{\Psi^{(n)}}\sim J^n \ket{\emptyset},
\end{equation}
where we introduced the operator $J$ which creates antipodal doublets:
\begin{equation}
  J=\sum_{j=1}^L  \sigma^+_j \sigma^+_{j+L}\,.
\end{equation}
Similar operators appeared in the study of quantum many body scars, see for example \cite{hfrag-review,scars-review-2}.

The states $\ket{\Psi^{(n)}}$ can be seen as a generalization of the Dicke states  to the antipodal setting. The
dimension of the space spanned by these vectors is $L+1$. 

We can also construct an alternative basis in the same space by starting with a different parametrization. We define
\begin{equation}
  \ket{\tilde\Psi^{(n)}}=\left.\left(\frac{d}{d\kappa}\right)^n\ket{\tilde\Psi(\kappa)}\right|_{\kappa=0},\qquad
  \ket{\tilde\Psi(\kappa)}=\otimes_{j=1}^{L}  \ket{\tilde\psi(\kappa)}_{j,j+L},
\end{equation}
where now 
\begin{equation}
  \ket{\tilde \psi(\kappa)}=(1+\kappa)\ket{0}\otimes \ket{0}+\kappa (\ket{1}\otimes \ket{1})\,.
\end{equation}
The states $\ket{\tilde\Psi^{(n)}}$ consist of $n$ symmetrically distributed antipodal Bell pairs that are added on top of
the ferromagnetic product state $\otimes_{j=1}^{2L}\ket{0}$.

The states $\ket{\Psi^{(n)}}$  and $\ket{\tilde\Psi^{(n)}}$ do not have volume law entanglement, instead the growth of the
entanglement if logarithmic if we keep $n$ fixed. We discuss the half-chain entangelment of the states $\ket{\Psi^{(n)}}$ in
detail. If we take the states $\ket{\Psi^{(n)}}$ (with the appropriate normalization added), and compute the half-chain
density matrices in the computational basis, then we will find that they are diagonal: any half-chain configuration is
coupled with precisely one configuration on the other half.  The state
$\ket{\Psi^{(n)}}$ has $n$ excitations $\ket{1}$ on the half-chain, therefore it has a total number of
\begin{equation}
  L \choose{n}
\end{equation}
non-zero components in the given basis. All the non-zero diagonal entries of the reduced density matrix are equal. It
follows that the half chain entropy is
\begin{equation}
  S=\log {{L}\choose{n}}
\end{equation}
Taking $L \gg n$ while keeping $n$ fixed we can approximate this by
\begin{equation}
  S\approx n\log(L)- \log(n!)
\end{equation}
We thus observe logarithmic growth in $L$ as expected.

The states $\ket{\Psi^{(n)}}$  and $\ket{\tilde\Psi^{(n)}}$  that we obtained  are analogous to the zero-energy states treated in the recent work
\cite{crosscap-tunable}, 
although their derivation and the mechanism for the annihilation of the states by the Hamiltonian is slightly different. In
\cite{crosscap-tunable} the key step 
towards the annihilation was to introduce alternating signs for the couplings in a specific Hamiltonian, and to have $L$ being
odd. In contrast, our methods apply for a set of coupling constants satisfying \eqref{halfp}, and the
annihilation follows from the condition \eqref{cond}. We will return to the construction of the work
\cite{crosscap-tunable} also in Section \ref{sec:integr}.

\subsection{The XX model}

Here we show the equivalence of the XY model defined above and the usual XX model, given that the total length of the
chain is a multiple of 4. Let us take the one site operator
\begin{equation}
  V=
  \begin{pmatrix}
    1 & 0 \\ 0 & i
  \end{pmatrix}
\end{equation}
and define the operator product
\begin{equation}
  \label{SSS}
  \VV=\prod_{j=1}^{2L} (V_j)^j
\end{equation}
Similarity transformation with $\VV$ leads to
\begin{equation}
  \VV^{-1}\sigma_j^+\VV=i^j\sigma^+_j,\qquad   \VV^{-1}\sigma_j^-\VV=(-i)^j\sigma^-_j\,.
\end{equation}
This leads to the relation
\begin{equation}
  \VV H_{XX}\VV^{-1}=H_{XY}\,,
\end{equation}
and the consequence that the states
\begin{equation}
  \ket{\tilde\Psi(c)} \equiv  \VV^{-1} \ket{\Psi(c)}
\end{equation}
are null-vectors of the XX model Hamiltonian.

Direct computation gives
\begin{equation}
  \label{XXcrosscap}
\ket{\tilde\Psi(c)}= \VV\ket{\Psi(c)}=  \otimes_{j=1}^L  \Big(\ket{0}_j\ket{0}_{j+L}+(-i)^{L} c \cdot (-1)^j  \ket{1}_j\ket{1}_{j+L}\Big)
\end{equation}
We have thus obtained crosscap states that are periodically modulated, in this case with period 2. 
Expanding these states in $c$ we
obtain the set of states
\begin{equation}
  \ket{\tilde\Psi^{(n)}}\sim \tilde J^n\ket{\emptyset},
\end{equation}
where now
\begin{equation}
 \tilde J= \sum_{j=1}^L   (-1)^j \sigma^+_j \sigma^+_{j+L}\,.
\end{equation}

\subsection{The Bariev model}

Let us now take a Hamiltonian with three-site density given by
\begin{equation}
  h=(X_1Y_3-Y_1X_3)(1+uZ_2),
  \label{eq:bariev-a}
\end{equation}
where $u$ is a real coupling constant. This model is closely related to the Bariev model \cite{bariev-model}, which is
an integrable spin chain that can also be understood as a zig-zag spin ladder model. 

The Hamiltonian density is anti-symmetric under transposition, and it preserves the global charge $Q$ given in \eqref{Sz}, therefore we
can choose the matrix $A$ to be diagonal with arbitrary entries. This implies that all the special states $\ket{\Psi(c)}$,
$\ket{\tilde\Psi(c)}$ that we constructed above, and also their derivatives are null-vectors of the Hamiltonian at hand.

Let us now assume that $L$ is divisible by 4, which means that the total length of the chain is a multiple of 8. We
define a twist operator via the formula \eqref{SSS} but now with the choice
\begin{equation}
  V=
  \begin{pmatrix}
    1 & 0 \\ 0 & \omega
  \end{pmatrix}
\end{equation}
where $\omega=e^{i\pi/4}$. Similar to the case of the XY and XX models now we obtain
that the Hamiltonian is transformed to a new one with operator density
\begin{equation}
  h=(X_1X_3+Y_1Y_3)(1+uZ_2)
  \label{eq:bariev}
\end{equation}
This is the Hamiltonian of the Bariev model  \cite{bariev-model}. As an effect, the zero energy eigenstates are transformed to
\begin{equation}
\VV\ket{\Psi(c)}=\otimes_{j=1}^L  \Big(\ket{1}_j\ket{1}_{j+L}+(-1)^{L/4}i^{j} c\ket{2}_j\ket{2}_{j+L}\Big)
\end{equation}
We obtain crosscap states that are modulated with period 4. If the length of the chain is a multiple of 8, then these
states are exact eigenstates of the Bariev model.  It appears that these crosscap states 
were not noticed in previous literature.

In contrast to the XY and XX models, the Bariev model is not free. However, all eigenstates of the Bariev model can be
constructed using the Bethe Ansatz 
\cite{bariev-model}, more precisely using the so-called nested Bethe Ansatz.
It is interesting that the existence of these special crosscap nullvectors is not limited to free models. The Bariev
model is a generic integrable model with regards to its integrability structures, because the scattering factors
and their Bethe equations are tunable \cite{bariev-model}.

The null space of the Bariev model is exponentially degenerate in the system size if the size is a multiple of 4. This can be proven by noting the
existence of four operators that anti-commute with the Hamiltonian. The first operator can be chosen as
\begin{equation}
  \DD=(X_1Z_2Y_3Z_4)\cdot (X_5Z_6Y_7Z_8)\cdots\,,
\end{equation}
and the others are obtained by cyclic shifts. As discussed in Section \ref{sub:exp}, the existence of $\DD$ leads to the exponential degeneracy of the null space.

In the case of the anti-symmetric model given by \eqref{eq:bariev-a} we find an anti-commuting operator for every
$L$, for example
\begin{equation}
  \DD=(X_1Z_2X_3Z_4)\cdot (X_5Z_6X_7Z_8)\cdots \,.
\end{equation}
Thus the null space will be exponentially degenerate in system size $L$ for all volumes.

In the Conclusions we will return to the connections between Bethe Ansatz solvability and the exponential degeneracies
of the nullspace.

\subsection{The folded XXZ model}

The folded XXZ model is an interacting model which can be derived using a special large anisotropy limit of the
Heisenberg XXZ chain \cite{folded1,sajat-folded}. It apppeared in multiple papers independently
\cite{folded0a,folded0b,folded1,sajat-folded} (see also \cite{folded-XXZ-stochastic-1}). Arguably it is one of the
simplest interacting integrable models, due to the simplicity of its scattering matrix \cite{sajat-folded}. The model
was realized recently in experiments with Rydberg atoms \cite{rydberg-folded-exp1,rydberg-folded-exp2}.

Here we show that the folded XXZ model also has exact crosscap zero-modes. We start with the following Hamiltonian
density with interaction range $4$:
\begin{equation}
  h=\frac{1+Z_1Z_4}{2} (X_2Y_3-Y_2X_3)\,.
\end{equation}
This Hamiltonian density is also anti-symmetric and it conserves the global charge $Q$. Accordingly, the crosscap states
$\ket{\Psi(c)}$ \eqref{eq:Psic} are exact zero-modes.

As in the case of the XY model, we can perform a similarity transformation with the operator $\VV$ given in \eqref{SSS}, which
leads to the Hamiltonian density of the folded XXZ model given by
\begin{equation}
  \label{foldedh}
  h=\frac{1+Z_1Z_4}{2} (X_2X_3+Y_2Y_3)\,.
\end{equation}
From this transformation we conclude that if the total length of the chain is a multiple of 4, then the states $\ket{\tilde \Psi(c)}$ \eqref{XXcrosscap} are exact
zero-energy eigenstates of the folded XXZ model.

\subsection{Non-integrable models}

It is crucial that even though the XY model defined by \eqref{hxy} is exactly solvable by free fermions, the computations
leading to the special crosscap states did not make use of the exact solvability.
In fact, we can find multiple non-integrable Hamiltonians, whose operator density nevertheless satisfies the condition \eqref{cond} in the same way. The easiest way towards non-integrability is to add
various types of controls on the hopping terms. Note that adding new interaction terms (for example diagonal
$Z$-$Z$ interactions as in the XXZ Heisenberg chain) would typically spoil the condition \eqref{cond}. This is the
reason why we explore adding controls for the hopping terms.

The simplest example is perhaps the three site density
\begin{equation}
  h=(1+uZ_1)(X_2Y_3-Y_2X_3),
\end{equation}
where $u$ is a real coupling constant. This Hamiltonian density satisfies the condition \eqref{cond} with the same
$A$-matrices, and the model is not integrable for $u \ne 0$.

We can once again perform a similarity transformation using $\VV$, which leads to the Hamiltonian density
\begin{equation}
  h=(1+uZ_1)(X_2X_3+Y_2Y_3).
\end{equation}
At the special points $u=\pm 1$ this coincides with the model treated in \cite{maksym-east}. That special model has
Hilbert space fragmentation.

A non-integrable model was considered in \cite{dhar-alternative-to-folded}, whose Hamiltonian is similar to that of the
folded XXZ model \eqref{foldedh}. The difference lies in the controls for the hopping of the particles. The Hamiltonian density is given
by
\begin{equation}
  \label{hfr}
  h_{fr}=\frac{1-Z_1Z_4}{2} (X_2X_3+Y_2Y_3)
\end{equation}
It was shown in \cite{dhar-alternative-to-folded} that this model also has Hilbert space fragmentation. It follows from
our results that the 
modulated crosscap states $|\tilde \Psi(c)\rangle$ are eigenstates whenever the volume is a multiple of 4.

As an alternative to the controlled hopping models, we can also construct models with correlated hopping.
For example consider the four site density given by
\begin{equation}
  h=i(\ket{0110}\bra{1001}-\ket{1001}\bra{0110})\,.
\end{equation}
This model is similar to the model treated in \cite{moudgalya2019thermalization}, the only difference being the
coefficients of the two correlated hopping amplitudes. This model belongs to the class of dipole-conserving models
\cite{tibor-fragment}, therefore it has Hilbert space fragmentation. We conjecture that similar to the model of
\cite{moudgalya2019thermalization} it is generally non-integrable, although it can have selected sectors in which it is
integrable. This Hamiltonian density is anti-symmetric and it conserves the global charge $Q$, therefore the states
$\ket{\Psi(c)}$ are exact zero-energy eigenstates.

\subsection{A model with a potentially linear nullspace degeneracy}

All previous examples are such that they have tunable crosscap eigenstates, and their null space is exponentially
degenerate. It is an interesting question whether there exist models where the degeneracy is only linear. We remind that
at least linear degeneracy is required for existence of 
the tunable crosscap states.

We consider the model with the mixed four site interaction terms
\begin{equation}
  \label{hmix}
  h_{mix}=\alpha(X_1Y_2-Y_1X_2)(X_3X_4+Y_3Y_4)+(1+\beta Z_1)(X_2Y_3-Y_2X_3)(1+\gamma Z_4)\,.
\end{equation}
Here $\alpha, \beta, \gamma$ are arbitrary real coupling constants. This operator is anti-symmetric under transposition and it commutes with
the global charge $Q$, therefore the states $\ket{\Psi(c)}$  are 
exact zero energy eigenstates.

We numerically computed the nullspace degeneracy for this model and for a selection of other models from the previous
examples. The numerical data for the other models is presented in the Table \ref{tab:zerosp}. The degeneracy of the nullspace in these models is exponential in system size. On the other hand, we observe that the nullspace degeneracy in the current
model (Table \ref{tab:zerosp-lin}) is much smaller than in the other examples, and it is compatible with linear
growth.

\begin{table}[h!]
  \centering
  \begin{tabular}{|c|c|c|c|c|}
    \hline
    $2L$   & 8 & 10 & 12 & 14 \\
    \hline
    \hline
    XY model (eq. \eqref{hxy}), & 60  & 144  & 386  & 862   \\
       \hline
      Bariev model (eq. \eqref{eq:bariev-a}) & 72 & 92 & 330  &  408  \\
       \hline
    folded XXZ model  (eq. \eqref{foldedh})& 96 & 220 & 776 & 1972 \\
    \hline
    non-integrable fragmented model  (eq. \eqref{hfr}) & 102 & 252 & 812 & 1972 \\
    \hline
  \end{tabular}
  \caption{Dimension of the zero space for selected integrable and non-integrable models.}
  \label{tab:zerosp}
\end{table}

\begin{table}[h!]
  \centering
  \begin{tabular}{|c|c|c|c|c|c|c|c|c|}
    \hline
    $2L$   & 8 & 10 & 12 & 14 & 16 & 18 & 20 & 22\\
    \hline
    \hline
    degeneracy & 16 & 16 & 22 & 32 & 48 & 52 & 66 & 84\\
    \hline
  \end{tabular}
  \caption{Dimension of the zero space for the mixed term model \eqref{hmix}.}
  \label{tab:zerosp-lin}
\end{table}

\section{Crosscap states for higher charges of integrable models}

\label{sec:integr}

In this Section we discuss connections to the integrable states of integrable spin chains. 
We focus on integrable Hamiltonians with nearest neighbour interactions. 

Integrable spin chains possess a tower of commuting conserved charges.
Let us denote the charges generally as
$Q^{(\alpha)}$, where $\alpha$ is an integer denoting the range of the corresponding operator density. In the typical
scenario there is only one charge for a given $\alpha$, and with these
conventions the Hamiltonian can be identified with $Q^{(2)}$. If there is a $U(1)$ charge in the model then it is
described by $Q^{(1)}$. In the typical scenario the charges can be chosen such that $Q^{(\alpha)}$ has eigenvalue
$(-1)^\alpha$ under spatial reflection.

In the theory of integrable spin chains a special role is played by the so-called integrable states. They are not
eigenstates of the Hamiltonian, instead they are defined by the following annihilation relations \cite{sajat-integrable-quenches}:
\begin{equation}
  Q^{(2j+1)}\ket{\Psi}=0\,,
  \label{intdef}
\end{equation}
i. e., they are annihilated by all the charges that are odd under spatial reflection. 
The common null space of these odd charges is exponentially degenerate, as we explain shortly. 

The actual eigenstates of the Hamiltonian can often be constructed using the Bethe Ansatz
\cite{gaudin-book-forditas}. Eigenstates are then characterized by a set of rapidities
$\ket{\boldsymbol \lambda} = \ket{\lambda_1,\dots,\lambda_N}$, which parametrize the lattice 
quasi-momenta of the interacting excitations. In models with higher rank symmetries there can be multiple sets of
rapidities.

We now discuss which eigenstates can have a non-zero overlap with the integrable initial states $\ket{\Psi}$
  defined by \eqref{intdef}. 
The eigenstates of the Hamiltonian are eigenstates of all the charges $Q^{(n)}$, therefore, a non-zero overlap $\langle \boldsymbol \lambda | \Psi \rangle$ is obtained only if $Q^{(2j+1)}\ket{\boldsymbol \lambda}=0$ for all the odd charges. The eigenvalues of the charges are sums of single particle eigenvalues,
\begin{equation}
  Q^{(n)}\ket{\boldsymbol \lambda} = \sum_{j=1}^{N} \tilde q^{(n)}(\lambda_j) \ket{\boldsymbol \lambda}\,,
\end{equation}
and the single particle eigenvalues $\tilde q^{(2j+1)}(\lambda)$ of the odd charges are odd functions in $\lambda$.
It follows that $Q^{(2j+1)}\ket{\boldsymbol \lambda}=0$ is guaranteed if the set of rapidities is reflection symmetric: i.e., the rapidities
 come in pairs $(\lambda_k,-\lambda_k)$ or they have some special value such as $\lambda_k=0$. For a
discussion of the pair structure in models with higher rank symmetries we refer to \cite{tamas-all-twosite-overlaps}.  If the pairing constraint (or constraints) are imposed, there still remains an exponentially large set of
allowed eigenstates, whose number is given by half of the usual thermodynamic Yang-Yang entropy \cite{gaudin-book-forditas}. This explains why the
zero space of the odd charges is exponentially degenerate. 
 
Despite this exponential degeneracy of the null space,
it is still non-trivial that
one can find concrete integrable states $\ket{\Psi}$ in this space which can be constructed using local rules only. This is
possible due to special algebraic structures in the theory of boundary integrable models \cite{sajat-integrable-quenches}.

The simplest examples of locally constructed integrable states are two-site states of the form
\begin{equation}
  \ket{\Psi}=\otimes_{j=1}^L  \ket{\psi}_{2j-1,2j},
\end{equation}
where $\ket{\psi}$ is some dimer state. In the concrete case of the Heisenberg XXZ chain every two-site state leads to
an integrable state \cite{sajat-integrable-quenches}, whereas for higher rank models there are restrictions on the
possibilities \cite{tamas-all-twosite-overlaps}.  It is known that certain special Matrix Product States (MPS) are also integrable
states \cite{sajat-mps,tamas-all-mps-overlaps}. Integrable states were used as scars in selected non-integrable
perturbations of integrable models 
\cite{integrable-scars-1,integrable-scars-2}.

Another class of explicitly known integrable states are the crosscap states. They were classified in
\cite{crosscap-tamas} for a large family of models.

We can use these results to arrive at new examples of tunable crosscap null states. A significant difference from previous examples is that the integrable crosscap states are zero modes of the odd conserved charges of the
models. Therefore, we proceed by regarding the shortest higher charge $Q^{(3)}$ as the ``Hamiltonian''. This way we
can find many new examples of crosscap null states. In a concrete case we find states that 
are closely related to the construction of \cite{crosscap-tunable}.

\subsection{Crosscap states for the $SU(2)$ symmetric XXX model}

In the case of the XXZ Heisenberg spin chains, explicit expressions for higher charges can be found for example in
\cite{GM-higher-conserved-XXZ} or alternatively in the more recent works \cite{xyz-all-charges,nienhuis-xxz}. The
simplest case is the $SU(2)$-symmetric XXX model, whose Hamiltonian is
\begin{equation}
  H=Q^{(2)}=\sum_j  {\bf S}_j\cdot  {\bf S}_{j+1}\,,
\end{equation}
where ${\bf S}$ is a vector of operators that we choose as $(X,Y,Z)$. The first odd charge is given by 
\begin{equation}
  Q^{(3)}=\sum_j  q^{(3)}(j)=\sum_j {\bf S}_j\cdot ( {\bf S}_{j+1}\times   {\bf S}_{j+2})\,.
\end{equation}
Alternatively, the charge can be expressed using SWAP gates as
\begin{equation}
  q^{(3)}(1)=2i(P_{1,2}P_{2,3}-P_{2,3}P_{1,2}),
\end{equation}
where $P_{j,k}$ swaps the local Hilbert spaces at sites $j$ and $k$.

For the model where the Hamiltonian is taken to be $Q^{(3)}$, every crosscap state $\ket{\Psi}$ of the form \eqref{Psi} is a zero energy eigenstate. This can be proven directly using our condition
\eqref{cond}: the local density $h\equiv q^{(3)}$  is anti-symmetric with respect to transposition, and it commutes with
$\AAA^{(3)}=A\otimes A\otimes A$ for any matrix $A$.

This implies that we can construct a three-parameter family of crosscap states. For example we can choose the
parametrization
\begin{equation}
  A=
  \begin{pmatrix}
    1+a & b+ic \\  b-ic & -a
  \end{pmatrix}
\end{equation}
and construct a three parameter family of crosscap zeromodes $\ket{\Psi(a,b,c)}$.  At $a=b=c=0$ the state is the
ferromagnetic reference state $\otimes_{j=1}^{2L}\ket{0}$. Taking derivatives with respect to $a, b, c$ around $a=b=c=0$
we obtain states where various numbers of three different anti-podal Bell pairs are distributed symmetrically along
the chain, added on top of the ferromagnetic reference state. 

Alternatively, we can start from a two-parameter set of states with
\begin{equation}
  A=
  \begin{pmatrix}
    1 & b+ic \\  b-ic & 1
  \end{pmatrix}.
\end{equation}
At $b=c=0$ this gives a homogeneous crosscap state, namely the ``identity crosscap state'' $\ket{\Psi_0}$. Taking
derivatives with respect to $b$ and $c$ we obtain examples of the ``symmetric tensor scars'' introduced in the recent work
\cite{crosscap-tunable}. The entanglement properties of those states were studied in detail in the Appendix of
\cite{crosscap-tunable}. 

\section{Conclusions}

\label{sec:concl}

We have shown that crosscap states with tunable entanglement  appear in various sorts of local spin chain models, if the
condition \eqref{cond} is satisfied. We gave multiple examples of Hamiltonians, both integrable and non-integrable,
which satisfy this constraint.
Our integrable examples include models solvable by free fermions, but also genuinely interacting models such as the
Bariev model and the folded XXZ model. Despite the large body of literature devoted to integrable models and in particular to
free fermionic models, it appears that the presence of these special eigenstates was not noticed previously. Exceptions are those crosscap states that appeared as {\it integrable states} of an integrable model
\cite{crosscap-shota,crosscap-tamas}. However, these kinds of crosscap states are zero modes of a subset of the higher
conserved charges.

Integrable models are often solvable by the Bethe Ansatz. Generally it is believed that if the Bethe Ansatz is
applicable in a certain model, then it reproduces all the eigenstates.
On the other hand, the Bethe states are not a natural basis to describe the unusual entanglement
pattern of the crosscap states, because the Bethe Ansatz wave functions describe scattering processes of \emph{local}
excitations. Therefore, it is 
natural to suspect that whenever there are exact crosscap zeromodes in an integrable model, then the nullspace
should be exponentially degenerate. This would allow the crosscap states to be reconstructed from linear combinations of an
exponential number of Bethe states.  All our examples of integrable models do in fact have this exponential degeneracy.

In contrast, we also found a non-integrable model with 4-site interactions, where the tunable crosscap zero states exist,
but the zero space degeneracy is much smaller and it appears to be compatible with a linear growth. We remind that
a linear growth is always needed to accommodate a one-parameter family of crosscap states. 

The states $\ket{\Psi^{(n)}}$ that we derived in Section \ref{sec:XY} have definite magnetization, and their reduced density
matrices for segments not bigger than half of the chain are proportional to the identity matrix within the subspace with
the given magnetization. This means that the states simulate the infinite temperature ensemble with fixed
magnetization. This way they are natural extensions of the maximally entangled crosscap states, which simulate the
un-constrained infinite temperature ensemble \cite{yoneta-chiba,yoneta-thermal-pure}.

\vspace{1cm}
{\bf Acknowledgments} 

\medskip

We are thankful to Ajit C. Balram, Tam\'as Gombor, Hosho Katsura, Sashikanta Mohapatra, Olexei I. Motrunich,
Sanjay Moudgalya, Bhaskar Mukherjee, Istv\'an Vona and Yasushi
Yoneta for useful discussions. We thank G\'abor Tak\'acs for allowing us to use his high-performance computing cluster.
This research was supported by the
Hungarian National Research, Development and Innovation Office, NKFIH Grant No. K-145904 and the 
NKFIH excellence grant TKP2021-NKTA-64.

\bigskip


\providecommand{\href}[2]{#2}\begingroup\raggedright\endgroup

\end{document}